\documentclass[;11pt]{article}
\usepackage[utf8]{inputenc}
\usepackage[english]{babel}
\usepackage{fullpage}
\usepackage{xspace,framed}
\usepackage{bbm}
\usepackage{mathrsfs}
\usepackage{latexsym}
\usepackage{lmodern}

\usepackage{amsmath,amsthm,amssymb}
\newcommand{\remove}[1]{}

\newtheorem{thm}{Theorem}[section]

\newtheorem{proposition}[thm]{Proposition}
\newtheorem{lemma}[thm]{Lemma}
\newtheorem{dfn}[thm]{Definition}

\newtheorem{cor}[thm]{Corollary}

\def\_{\,\,\,\,\,}

\newcommand{\B}{\{0,1\}}

\renewcommand{\O}{O_k} 

\newcommand{\eps}{\epsilon}

\newcommand{\rkcomp}{$(r,k)$-compatible\xspace}
\newcommand{\rkcons}{$(r,k)$-consistent\xspace}

\newcommand{\poly}{n^{O(1)}}

\def\MMC{{\mathcal M(\C^k \to \C^n,d)}}
\def\MMC2{{\mathcal M(\C^k \to \C^k,d)}}

\newcommand{\A}{{\cal A}}
\renewcommand{\B}{{\cal B}}
\newcommand{\F}{{\cal F}}
\renewcommand{\P}{{\cal P}}

\newcommand{\rpaths}{$r$-\textrm{simple} paths\xspace}
\newcommand{\rkpath}{$r$-\textrm{simple} $k$-path\xspace}
\newcommand{\rsimp}{\textsc{$r$-Simple $k$-Path}\xspace}
\newcommand{\rpacking}{\textsc{$(r,p,q)$-Packing}\xspace}
\newcommand{\rset}{$r$-\textrm{set}\xspace}
\newcommand{\rsets}{$r$-\textrm{set}s\xspace}
\newcommand{\rkset}{$(r,k)$-\textrm{set}\xspace}
\newcommand{\rksets}{$(r,k)$-\textrm{set}s\xspace}
\newcommand{\rksep}{$(r,k)$-separator\xspace}
\newcommand{\sepsize}{\O(r^{6k/r}\cdot 2^{O(k/r)}\cdot \log n)}
\newcommand{\repsize} {\O(r^{6k/r}\cdot 2^{O(k/r)}\cdot \log n)}
\newcommand{\repsiz} {r^{6k/r}\cdot 2^{O(k/r)}\cdot \log n}

\newcommand{\repcomputenopoly}{r^{12k/r}\cdot 2^{O(k/r)}\cdot  \log n} 
\newcommand{\repcomputenon}{r^{12k/r}\cdot 2^{O(k/r)}} 

\newcommand{\mondetectcomputenon}{r^{18k/r}\cdot 2^{O(k/r)}} 

\newcommand{\sepcompute}{r^{6k/r}\cdot 2^{O(k/r)}\cdot n\log n} 
\newcommand{\qset}{$q$-\textrm{set}\xspace}
\newcommand{\qsets}{$q$-\textrm{set}s\xspace}
\newcommand{\rpack}{$r$-packing\xspace}
\newcommand{\rpacks}{$r$-packings\xspace}
\newcommand{\comp}[1]{\overline{#1}}
\newcommand{\trim}[1]{\mathrm{Trim}_{\F}(#1)}
\newcommand{\rep}[1]{\hat{#1}}
\newcommand{\lops}[2]{$(n,#1,#2)$-lopsided universal set\xspace}
\newcommand{\lopss}[2]{$(n,#1,#2)$-lopsided universal sets\xspace}

\newcommand{\rkmon}{$(r,k)$-monomial\xspace}
\newcommand{\rmondetect}{\textsc{$(r,k)$-Monomial Detection}\xspace}
\newcommand{\rmon}{$r$-monomial\xspace}
\newcommand{\rmons}{$r$-monomials\xspace}
\newcommand{\rr}{[r]_0}
\newcommand{\rkmons}{$(r,k)$-monomials\xspace}
\newcommand{\wt}{\mathbf{w}}

 \newcommand{\pr}{\mathbb{P}}

\newcommand{\set}[1]{\{#1\}}

\def\E{{\mathbb{E}}}

\def\Z{{\mathbb{Z}}}

\renewcommand{\H}{{\cal H}}

\newcommand{\oftypenon}[2]{\set{\,{#1}\to{#2}\,}}

\newcommand{\oftype}[2]{\subseteq \oftypenon{#1}{#2}}

\newcommand{\defparproblemu}[4]{
  \vspace{1mm}
\noindent\fbox{
  \begin{minipage}{0.97\textwidth}
  #1 \\[0.2cm]
  {\bf{Input:}} #2  \\
  {\bf{Parameters:}} #3 \\
  {\bf{Question:}} #4
  \end{minipage}
  }
  \vspace{1mm}
}

\title{Fast Algorithms for Parameterized Problems with Relaxed Disjointness Constraints\thanks{The research leading to these results has received funding from the European Community's Seventh Framework Programme (FP7/2007-2013) under grant agreement number 257575 and 240258. Mi. Pilipczuk is currently holding a post-doc position at Warsaw Center of Mathematics and Computer Science and is supported by Polish National Science Centre grant DEC-2013/11/D/ST6/03073.}}

\author{Ariel Gabizon\thanks{Department of Computer Science, Technion, Israel, \texttt{ariel.gabizon@gmail.com}.}
\and Daniel Lokshtanov\thanks{Department of Informatics, University of Bergen, Norway, \texttt{daniello@ii.uib.no}.}
\and Micha\l{} Pilipczuk\thanks{Institute of Informatics, University of Warsaw, Poland, \texttt{michal.pilipczuk@mimuw.edu.pl}.}
}

\begin{document}
\maketitle{}
\begin{abstract}
In parameterized complexity, it is a natural idea to consider different generalizations of classic problems. Usually, such generalization are obtained by introducing a ``relaxation'' variable, where the original problem corresponds to setting this variable to a constant value. For instance, the problem of packing sets of size at most $p$ into a given universe generalizes the {\sc{Maximum Matching}} problem, which is recovered by taking $p=2$. Most often, the complexity of the problem increases with the relaxation variable, but very recently Abasi et al. have given a surprising example of a problem --- $r$-Simple $k$-Path --- that can be solved by a randomized algorithm with running time $O^*(2^{O(k \frac{\log r}{r})})$. That is, the complexity of the problem decreases with $r$.

In this paper we pursue further the direction sketched by Abasi et al. Our main contribution is a derandomization tool that provides a deterministic counterpart of the main technical result of Abasi et al.: the $O^*(2^{O(k \frac{\log r}{r})})$ algorithm for $(r,k)$-Monomial Detection, which is the problem of finding a monomial of total degree $k$ and individual degrees at most $r$ in a polynomial given as an arithmetic circuit. Our technique works for a large class of circuits, and in particular it can be used to derandomize the result of Abasi et al. for $r$-Simple $k$-Path. On our way to this result we introduce the notion of representative sets for multisets, which may be of independent interest.

Finally, we give two more examples of problems that were already studied in the literature, where the same relaxation phenomenon happens. The first one is a natural relaxation of the {\sc{Set Packing}} problem, where we allow the packed sets to overlap at each element at most $r$ times. The second one is {\sc{Degree Bounded Spanning Tree}}, where we seek for a spanning tree of the graph with a small maximum degree.
\end{abstract}

\section{Introduction}
Many of the combinatorial optimization problems studied in theoretical computer science are idealized mathematical models of real-world problems. When the simplest model is well-understood, it can be enriched to better capture the real-world problem one actually wants to solve. Thus it comes as no surprise that many of the well-studied computational problems generalize each other: the {\sc Constraint Satisfaction Problem} generalizes {\sc Satisfiability}, the problem of finding a spanning tree with maximum degree at most $d$ generalizes {\sc Hamiltonian Path}, while the problem of packing sets of size $3$ generalizes packing sets of size $2$, also known as the {\sc Maximum Matching} problem.


By definition, the generalized problem is computationally harder than the original. However it is sometimes the case that the most difficult instances of the generalized problem are actually instances of the original problem. In other words,  the ``further away'' an instance of the generalized problem is from being an instance of the original, the easier the instance is.
Abasi et. al \cite{ABGH14} initiated the study of this phenomenon in parameterized complexity (we refer the reader to the textbooks~\cite{the-awesome-book,DowneyF13,FlumGrohebook,Niedermeierbook06} for an introduction to parameterized complexity). In particular, they study the {\sc $r$-Simple $k$-Path} problem. Here the input is a graph $G$, and integers $k$ and $r$, and the objective is to determine whether there is an $r$-simple $k$-path in $G$, where an $r$-simple $k$-path is a sequence $v_1, v_2, \ldots, v_k$ of vertices such that every pair of consecutive vertices is adjacent and no vertex of $G$ is repeated more than $r$ times in the sequence. Observe that for $r=1$ the problem is exactly the problem of finding a simple path of length $k$ in $G$. On the other hand, for $r=k$ the problem is easily solvable in polynomial time, as one just has to look for a walk in $G$ of length $k$. Thus, gradually increasing $r$ from $1$ to $k$ should provide a sequence of computational problems that become easier as $r$ increases.  Abasi et al. \cite{ABGH14} confirm this intuition by giving a randomized algorithm for {\sc $r$-Simple $k$-Path} with running time  $O(r^{2k/r}n^{O(1)})$.

In this paper we continue the investigation of algorithms for problems with a {\em relaxation parameter} $r$ that interpolates between an NP-hard and a polynomial time solvable problem. We show that in several interesting cases one can get a sequence of algorithms with better and  better running times as the relaxation parameter $r$ increases, essentially providing a smooth transition from the NP hard to the polynomial time solvable case.

Our main technical contribution is a new algorithm for the $(r,k)$-{\sc Monomial Detection} problem. Here the input is an arithmetic circuit $C$ that computes a polynomial $f$ of degree $k$ in $n$ variables $x_1,\ldots,x_n$. The task is to determine whether the polynomial $f$ has a monomial $\Pi_{i = 1}^n x_i^{a_i}$, such that $0 \leq a_i \leq r$ for every $i \leq n$. The main result of Abasi et al.~\cite{ABGH14} is a randomized algorithm for $(r,k)$-{\sc Monomial Detection} with running time   $O(r^{2k/r}|C|n^{O(1)})$, and their algorithm for {\sc $r$-Simple $k$-Path} is obtained using a reduction to $(r,k)$-{\sc Monomial Detection}. We give a {\em{deterministic}} algorithm for the problem with running time $r^{O(k/r)}|C|n^{O(1)}$ in the case when the circuit $C$ is {\em{non-canceling}}. Formally, this means that the circuit contains only variables at its leaves (i.e., no constants) and only addition and multiplication gates (i.e, no subtraction gates). Informally, all monomials of the polynomials computed at intermediate gates of $C$ contribute to the polynomial computed by $C$.

Comparing our algorithm with the algorithm of Abasi et al.~\cite{ABGH14}, our algorithm is slower by a constant factor in the exponent of $r$, and only works for non-cancelling circuits. However our algorithm is {\em deterministic} (while the one by Abasi et al. is randomized) and also works for the {\em weighted} variant of the problem, while the one by Abasi et al. does not. In the weighted variant each variable $x_i$ has a non-negative integer weight $w_i$, and the weight of a monomial $\Pi_{i = 1}^n x_i^{a_i}$ is defined as $\sum_{i = 1}^n w_ia_i$. The task is to determine whether there exists a monomial $\Pi_{i = 1}^n x_i^{a_i}$, such that $0 \leq a_i \leq r$ for every $i \leq n$, and if so, to return one of minimum weight.  As a direct consequence we obtain the first deterministic algorithm for {\sc $r$-Simple $k$-Path} with running time  $r^{O(k/r)}|C|n^{O(1)}$, and the first algorithm with such a running time for weighted  {\sc $r$-Simple $k$-Path}.

The significance of an in-depth study of $(r,k)$-{\sc Monomial Detection}, is that it is the natural ``relaxation parameter''-based generalization of the fundamental {\sc Multi-linear Monomial Detection} problem. The {\sc Multi-linear Monomial Detection} problem is simply $(r,k)$-{\sc Monomial Detection} with  $r=1$. A multitude of parameterized problems reduce to {\sc Multi-linear Monomial Detection}~\cite{BjorklundHKK10,the-awesome-book,FLPS14,K08,W09}. Thus, obtaining good algorithms $(r,k)$-{\sc Monomial Detection} is an important step towards efficient algorithms for relaxation parameter-variants of these problems. For some problems, such as {\sc $k$-Path}, efficient algorithms for the relaxation parameter variant (i.e {\sc $r$-Simple $k$-Path}) follow directly from the algorithms for $(r,k)$-{\sc Monomial Detection}. In this paper we give two more examples of fundamental problems for which efficient algorithms for $(r,k)$-{\sc Monomial Detection} lead to efficient algorithms for their ``relaxation parameter''-variant.

Our first example is the {\sc $(r,p,q)$-Packing} problem. Here the input is a family ${\cal F}$ of sets of size $q$ over a universe of size $n$, together with integers $r$ and $p$. The task is to find a subfamily ${\cal A} \subseteq {\cal F}$ of size at least $p$ such that every element of the universe is contained in at most $r$ sets in ${\cal A}$. Observe that {\sc $(r,p,q)$-Packing} is the relaxation parameter variant of the classic {\sc Set Packing} problem ({\sc $(r,p,q)$-Packing} with $r=1$). We give an algorithm for {\sc $(r,p,q)$-Packing} with running time $2^{O(pq \cdot \frac{\log r}{r})}|{\cal F}|n^{O(1)}$. For $r=1$ our algorithm matches the best known algorithm~\cite{BjorklundHKK10} for {\sc Set Packing}, up to constants in the exponent, and when $r$ grows our algorithm is significantly faster than $2^{pq}|{\cal F}|n^{O(1)}$. Just as for {\sc $r$-Simple $k$-Path}, our algorithm also works for weighted variants of the problem. We remark that  {\sc $(r,p,q)$-Packing}  was also studied by Fernau et al.~\cite{FernauLR15} from the perspective of kernelization.

Our second example is the {\sc Degree-Bounded Spanning Tree} problem. Here, we are given as input a graph $G$ and integer $d$, and the task is to determine whether $G$ has a spanning tree $T$ whose maximum degree does not exceed $d$. For $d = 2$ this problem is equivalent to {\sc Hamiltonian Path}, and hence the problem is NP-complete in general, but for $d=n-1$ it boils down to checking the connectedness of $G$. Thus, {\sc Degree-Bounded Spanning Tree} can be thought of as a relaxation parameter variant of {\sc Hamiltonian Path}. The problem has received significant attention in the field of approximation algorithms: there are classic results of Goemans~\cite{Goemans06} and of Singh and Lau~\cite{SinghL15} that give additive approximation algorithms for the problem and its weighted variant. From the point of view of exact algorithms, the currently fastest exact algorithm, working for any value of $d$, is due to Fomin et al.~\cite{FominGLS12} and has running time $O(2^{n+o(n)})$. In this work, we give a randomized algorithm for {\sc Degree-Bounded Spanning Tree} with running time $2^{O(n \frac{\log d}{d})}$, by reducing the problem to an instance of $(r,k)$-{\sc Monomial Detection}. Thus, our algorithm significantly outperforms the algorithm of Fomin et al.~\cite{FominGLS12} for all super-constant $d$, and runs in polynomial time for $d = \Omega(n)$. Interestingly, the instance of $(r,k)$-{\sc Monomial Detection} that we create crucially uses subtraction, since the constructed circuit computes the determinant of some matrix. Thus we are not able to apply our algorithm for non-cancelling circuits, and have to resort to the randomized algorithm of Abasi et al.~\cite{ABGH14} instead. Obtaining a deterministic algorithm for {\sc Degree-Bounded Spanning Tree} that would match the running time of our algorithm, or extending the result to the weighted setting, remains as an interesting open problem.

\paragraph{Our methods.}
The starting point for our algorithms is the notion of {\em representative sets}.
If ${\cal A}$ is a family of sets, with all sets in ${\cal A}$ having the same size $p$, we say that a subfamily ${\cal A}' \subseteq {\cal A}$ $q$-{\em represents} ${\cal A}$ if
for every set $B$ of size $q$, whenever there exists a set $A \in {\cal A}$ such that $A$ is disjoint from $B$, then there also exists a set $A' \in {\cal A}'$ such that $A'$ is disjoint from $B$.

Representative sets were defined by Monien~\cite{Monien85}, and were recently used to give efficient parameterized algorithms for a number of problems~\cite{FLPS14,FLS14,PinterSZ14,SZ14,Zehavi13,Zehavi14}, including $k$-{\sc Path}~\cite{FLPS14,FLS14,SZ14}, {\sc Set Packing}~\cite{SZ14,Zehavi13,Zehavi14} and {\sc Multi-linear Monomial Detection}~\cite{FLPS14}. It is therefore very tempting to try to use representative sets also for the relaxation parameter variants of these problems. However, it looks very hard to directly use representative sets in this setting. On a superficial level the difficulty lies in that representative sets are useful to guarantee {\em disjointedness}, while the solutions to the relaxation parameter variants of the considered problems may self-intersect up to $r$ times.

We overcome this difficulty by generalizing the notion of representative sets to multisets. When taking the union of two multisets $A$ and $B$, an element that appears $a$ times in $A$ and $b$ times in $B$ will appear $a+b$ times in the union $A + B$. Thus, if two regular sets $A$ and $B$ are viewed as multisets, they are disjoint if and only if no element appears more than once in $A+B$.  We can now relax the notion of disjointedness and require that no element appears more than $r$ times in $A+B$. Specifically, if ${\cal A}$ is a family of multisets, with all multisets in ${\cal A}$ having the same size $p$ (counting duplicates), we say that a subfamily ${\cal A}' \subseteq {\cal A}$ $q$-{\em represents} ${\cal A}$ if the following condition is satisfied. For every multiset $B$ of size $q$, whenever there exists an $A \in {\cal A}$ such that no element appears more than $r$ times in $A+B$, there also exists an $A' \in {\cal A}'$ such that no element appears more than $r$ times in $A' + B$. The majority of the technical effort in the paper is spent on proving that every family ${\cal A}$ of multisets has a relatively small $q$-representative family ${\cal A}'$ in this new setting, and to give an efficient algorithm to compute ${\cal A}'$ from ${\cal A}$. The formal statement of this result can be found in Corollary~\ref{cor:repset}.

On the way to develop our algorithm for computing representative sets of multisets, we give a new construction of a pseudo-random object called \emph{lopsided universal sets}. Informally speaking, an \emph{\lops{p}{q}} is a set of strings such that, when focusing on any $k\triangleq p+q$ locations, we see all patterns of hamming weight $p$.
These objects have been of interest for a while in mathematics and theoretical computer science
under the name \emph{Cover Free Families} (Cf. \cite{Bshouty14}).
We give, for the first time, an explicit construction of an \lops{p}{q} whose size is only polynomially larger than optimal for all $p$ and $q$.
See Theorem \ref{thm:lopsided_main} in Section \ref{sec:sepfam} for a formal statement.
Both our algorithm for computing representative sets of multisets, and the new construction of lopsided universal sets may be of independent interest.

\paragraph{Outline of the paper.} In Section~\ref{sec:prelim} we give the necessary definitions and set up notational conventions. In Section~\ref{sec:repset} we give our construction of representative sets for multisets. This construction requires an auxiliary tool called minimal separating families. The construction of representative sets for multisets in Section~\ref{sec:repset} assumes that an appropriate construction of minimal separating families is given as a black box, while the construction of minimal separating families is deferred to Section~\ref{sec:sepfam}. Our new construction of lopsided universal sets is a corollary of the construction of minimal separating families, and is also explained in Section~\ref{sec:sepfam}. In Section~\ref{sec:algapp} we use the new construction for representative sets of multisets to give efficient algorithms for  $(r,k)$-{\sc Monomial Detection}, {\sc $(r,p,q)$-Packing} and  {\sc $r$-Simple $k$-Path}. In Section~\ref{sec:kirchoff} we present our algorithm for {\sc Degree-Bounded Spanning Tree}. Finally, we conclude by discussing open problems and directions for future research in Section~\ref{sec:concl}.

\section{Preliminaries}\label{sec:prelim}

\paragraph*{Notation.} 
Throughout the paper, we use the notation $\O$ to hide $k^{O(1)}$ terms. We denote $[n]=\{1,2,\ldots,n\}$. For sets $A$ and $B$, by $\oftypenon{A}{B}$ we denote the set of all functions from $A$ to $B$. 
The notation $\triangleq$ is used to introduce new objects defined by formulas on the right hand side.



\paragraph*{Hashing families.} 
Recall that, for an integer $t\geq 1$, we say that a family of functions $\H\oftype{[n]}{[m]}$ is a \emph{$t$-perfect hash family},
if for every $C\subseteq [n]$ of size $|C| = t$ there is $f\in \H$ that is injective on $T$. Alon, Yuster and Zwick \cite{AYZ08} used a construction of Moni Naor
(based on ideas from Naor et al.~\cite{NSS95}) to hash a subset of size $t$ into a world of size $t^2$ using a very small set of functions:
\begin{thm}[\cite{AYZ08} based on Naor]\label{thm:hash_to_k^2}
For integers $1\leq t\leq n$, a $t$-perfect hash family $\H\oftype{[n]}{[t^2]}$
of size $t^{O(1)}\cdot \log n$ can be constructed in time $O(t^{O(1)}\cdot n\cdot \log n)$
\end{thm}

We will also use the following perfect hash family given by Naor, Schulman and Srinivasan \cite{NSS95}.
\begin{thm}[\cite{NSS95}]\label{thm:NSShash}
For integers $1\leq t\leq n$, a $t$-perfect hash family $\H\oftype{[k^2]}{[t]}$ of size $e^{t+ O(\log^2 t)}\cdot \log k$
can be constructed in time $O(e^{t+ O(\log^2 t)}\cdot k\cdot \log k)$.
\end{thm}

\paragraph*{Separating families.}
We will be interested in constructing families of perfect hash functions
that, in addition to being injective on a set $C$, have the property of sending another large set $D$ to
an output \emph{disjoint} from the image of $C$.
We call such a family of functions a \emph{separating family}.

\begin{dfn}[Separating family]\label{dfn:seperating}
Fix integers $t,k,s,n$ such that $1\leq t\leq n$.
For disjoint subsets $C,D\subseteq [n]$, we say that a function $h\colon [n] \to [s]$
\emph{separates} $C$ from $D$ if
\begin{itemize}
\item $h$ is injective on $C$; and
\item there are no collisions between $C$ and
$D$. That is, $h(C)\cap h(D)=\emptyset$.
\end{itemize}

A family of functions $\H\oftype{[n]}{[s]}$
is \emph{$(t,k,s)$-separating}
if for every disjoint subsets $C,D\subseteq [n]$
with $|C| = t $ and $|D| \leq k-t$, there is a function $h\in \H$
that separates $C$ from $D$.

We say that $\H$ is \emph{$(t,k,s)$-separating with probability $\gamma$}
if for any fixed $C$ and $D$ with sizes as above, a function $h$ chosen uniformly at random from $\H$ separates $C$ from $D$
with probability at least $\gamma$.
\end{dfn}

For us, the most important case of separating families is when the range size is $|C|+1$.
In this case we use the term \emph{minimal separating family}.
It will also be convenient to assume in the definition that $C$ is mapped
to the first $|C|$ elements in the range.
\begin{dfn}[Minimal separating family]\label{dfn:minseparating}
A family of functions $\H\oftype{[n]}{[t+1]}$
is \emph{$(t,k)$-minimal separating} if
for every disjoint subsets $C,D\subseteq [n]$
with $|C| =  t $ and $|D| \leq  k-t$, there is a function $h\in \H$ such that
\begin{itemize}
\item $h(C)=[t]$.
\item $h(D)\subseteq \{t+1\}$.
\end{itemize}
\end{dfn}

\section{Multiset separators and representative sets}\label{sec:repset}

The purpose of this section is to formally define and construct representative sets for multisets. We will use, as an auxiliary result, an efficient construction of a small separating family.

\begin{thm}\label{thm:sepratingmain}
Fix integers $n,t,k$ such that $1\leq t\leq \min(n,k)$. Then a $(t,k)$-minimal separating family of size $\O((k/t)^{2t}\cdot 2^{O(t)}\cdot \log n)$
can be constructed in time $\O((k/t)^{2t}\cdot 2^{O(t)}\cdot n\cdot \log n)$.
\end{thm}

We defer the proof of Theorem~\ref{thm:sepratingmain} to Section~\ref{sec:sepfam}, and now we shall use it as a blackbox in order to provide a construction of representative sets for multisets. The primary tool for this is what we call a \emph{multiset separator} (see Definition \ref{dfn:multiseparator}).
Informally, a multiset separator is a not too large set of `witnesses' for the fact
that two multisets of bounded size do not jointly contain too many repetitions per element.

\paragraph*{Notation for multisets.}
Fix integers integers $n,r,k\geq 1$.
We use $\rr$ to denote $\set{0,\ldots,r}$.
An \emph{\rset} is a multiset $A$ where each element of $[n]$
appears at most $r$ times.
It will be convenient to think of $A$ as a vector in $\rr^n$,
where $A_i$ denotes the number of times $i$ appears in $A$.
We denote by $|A|$ the number of elements in $A$ counting repetitions.
That is, $|A| = \sum_{i=1}^n A_i $.
We refer to $|A|$ as the \emph{size} of $A$.
An \emph{\rkset} is an \rset $A\in \rr^n$,
where the number of elements with repetitions is at most $k$.
That is, $|A| \leq k$. For two multisets $A,B$ over $[n]$, 

Fix \rsets $A,B\in \rr^n$.
We say that $A\leq B$ when
$A_i\leq B_i$ for all $i\in [n]$.
By $\comp{A}\in \rr^n$ we denote
the ``complement'' of \rset $A$,
that is, $\comp{A}_i = r-A_i$ for all $i\in [n]$. By $A+B$ we denote the ``union'' of $A$ and $B$, that is, $(A+B)_i=A_i+B_i$ for all $i\in [n]$.
Suppose now that $A$ and $B$ are \rksets.
We say that $A$ and $B$ are \emph{\rkcomp}
if $A+B$ is also an \rkset, and $|A+B|=k$.
That is, the total number of elements with repetitions in $A$ and $B$ together is $k$
and any specific element $i\in [n]$ appears in $A$ and $B$ together at most $r$ times.
With the notation above at hand, we can define the central object needed for our algorithms.
\begin{dfn}[Multiset separator]\label{dfn:multiseparator}
Let $\F$ be a family of \rsets.
We say that $\F$ is an \rksep if
for any \rksets $A,B \in \rr^n$
that are \rkcomp, there exists $F\in \F$
such that $A\leq  F\leq \comp{B}$.
\end{dfn}

\paragraph*{Construction of multiset separators.} The following theorem shows how an \rksep can be constructed from a minimal separating family.
\begin{thm}\label{thm:multisetsepmain}
Fix integers $n,r,k$ such that $1<r\leq k$,
and let $t \triangleq  \lfloor 2k/r \rfloor$.
Suppose a $(t,k)$-minimal separating family $\H\oftype{[n]}{[t+1]}$
can be constructed in time $f(r,k,n)$.
Then an \rksep $\F$ of size $|\H|\cdot (r+1)^t$
can be constructed in time $\O(f(r,k,\max(n,t))\cdot (r+1)^t)$.
\end{thm}
\begin{proof}
In the proof we will assume that $n\geq t$, since otherwise we can apply the same construction for $n$ increased to $t$, and at the end remove from each multiset of the obtained $\F$ all the elements from $[t]\setminus [n]$. Note that thus we apply the construction of a minimal separating family to the set of size $\max(n,t)$, rather than $n$. Also, observe that from the assumption that $r>1$ it follows that $t\leq k$.

Let $\H$ be the constructed $(t,k)$-minimal separating family of functions from $[n]$ to $[t+1]$. For each $h:[n]\to [t+1]$ in $\H$, and
for each $w=(w_1,\ldots,w_t)\in \rr^t$,
we construct the following \rset $F^{h,w}\in \rr^n$.
For all $j\in [t]$ and all $i\in [n]$ with $h(i)=j$, we put $F^{h,w}_i = w_j$.
For all $i\in [n]$ with $h(i)=t+1$, we put $F^{h,w}_i=r/2$.
Let $\F$ consist of all the constructed \rsets $F^{h,w}_i$.
Thus we have that
\[|\F|= |\H|\cdot (r+1)^{t},\]
and $\F$ clearly can be constructed in time as claimed in the theorem statement.
We are left with proving that $\F$ is indeed an \rksep.

Fix \rkcomp \rksets $A,B \in \rr^n$. Let $U$ be the set of all elements that appear in $A$ or in $B$, that is, $U = \set{i\in [n]\mid A_i>0\; \mathrm{or} \; B_i>0}$.
Denote by $C_0\subseteq U$ the sets of elements
in $A$ and $B$ that appear more than  $r/2$ times in one of the sets, that is, $C_0= \set{i\in [n]\mid A_i> r/2 \; \mathrm{or} \; B_i> r/2}$. Note that since $A$ and $B$ are \rksets, we have that
$|C_0|\leq \lfloor 2k/r \rfloor= t$. Let $C$ be a superset of $C_0$ of size exactly $t$, constructed by augmenting $C_0$ with arbitrary elements of $U\setminus C_0$ up to size $t$, and if there is not enough of them, then by additionally augmenting it with the remaining number of arbitrary elements of $[n]\setminus U$. Note that this is always possible since $t\leq n$. Let $D=U\setminus C$. Since $A$ and $B$ are \rkcomp, we have that $|U|\leq k$ and from the construction of $C$ it follows that $|D|\leq k-t$.

Therefore, there exists some $h\in \H$ that separates $C$ from $D$.
For $j\in [t]$, let $i_j$ be the unique element of $C$ mapped to $j$ under $h$.
Choose $w_j\in \rr$ such that $A_{i_j}\leq w_j \leq r-B_{i_j}$; such $w_j$ exists since $A$ and $B$ are \rkcomp.
Let $w=(w_1,\ldots,w_t)$.
We claim that
$A\leq F^{h,w}\leq \comp{B}$.
For $i\in C$, the choice of $w$ guarantees that
$A_{i}\leq F^{h,w}_i \leq \comp{B}_i$.
For $i\in D$ we have $F^{h,w}_i = \lfloor r/2\rfloor$,
and from the definition of $D$ we
have that $D\cap C_0=\emptyset$. So for such $i$ it holds that $A_i \leq \lfloor r/2\rfloor \leq \comp{B}_i$.
Finally, for $i\notin C\cup D$
we have that $A_i=0$ and $\comp{B}_i =r$, so surely
$A_i \leq F^{h,w}_i \leq \comp{B}_i$.
\end{proof}

By combining Theorems~\ref{thm:sepratingmain} and~\ref{thm:multisetsepmain} we immediately obtain the following construction.

\begin{cor}\label{cor:multisetsepmain}
Fix integers $n,r,k$ such that $1< r\leq k$. Then an \rksep $\F$ of size $\O(r^{6k/r}\cdot 2^{O(k/r)}\cdot \log n)$
can be constructed in time $\O(r^{6k/r}\cdot 2^{O(k/r)}\cdot n\cdot \log n)$
\end{cor}
\begin{proof}
Let $t\triangleq \lfloor 2k/r \rfloor$. By Theorem~\ref{thm:sepratingmain}, a $(t,k)$-minimal separating family $\H\oftype{[n]}{[t+1]}$ of size $\O((k/t)^{2t}\cdot 2^{O(t)}\cdot \log n)$
can be constructed in time $\O((k/t)^{2t}\cdot 2^{O(t)}\cdot n\cdot \log n)$ from Theorem \ref{thm:sepratingmain}.
Using this construction in Theorem \ref{thm:multisetsepmain}, we obtain an \rksep $\F$ such that
\[|\F|= |\H|\cdot (r+1)^{t} = \O(r^{6k/r}\cdot 2^{O(k/r)}\cdot \log n).\]
Moreover, from Theorem~\ref{thm:multisetsepmain} it follows that $\F$ can be constructed in time $\O(r^{6k/r}\cdot 2^{O(k/r)}\cdot n\cdot \log n)$.
\end{proof}

\paragraph*{Multisets over a weighted universe.} 
Before proceeding, we discuss the issue of how the considered multisets will be equipped with \emph{weights}.
For simplicity, we assume that the universe $\set{1,\ldots,n}$ is weighted, that is, each element $i\in \set{1,\ldots,n}$ is assigned an integer weight $\wt(i)$.
We define the weight of a multiset as the sum of the weights of its elements counting repetitions.
Formally, for $A\in \rr^n$ we have
\[\wt(A) = \sum_{i=1}^n A_i\cdot \wt(i).\]
Whenever we talk about a {\em{weighted family}} of multisets, we mean that the universe $[n]$ is equipped with a weight function and the weights of the multisets are defined as in the formula above.

Let us remark that the results to follow can be also extended to a more general setting where each multiset is assigned its own weight that is not directly related to its elements. However, for concreteness we now focus on the simpler case. We discuss briefly the generalization in Section~\ref{sec:algapp}, where we argue that our tools can be also used to solve the edge-weighted variant of \rsimp.

\paragraph*{Representative sets for multisets.} We are ready to define the notion of a representative set for a family of multisets.

\begin{dfn}[Representative sets for multisets]\label{dfn:representativeset}
Let $\P$ be a weighted family of \rksets. We say that a subfamily $\rep{\P}\subseteq \P$ \emph{represents} $\P$
if for every \rkset $Q$ the following holds.
If there exists some $P\in \P$ of weight $w$ that is \rkcomp
with $Q$, then there also exists some $P'\in \rep{P}$ of weight $w'\leq w$
that is \rkcomp with $Q$.
\end{dfn}


The following definition and lemma show that having an \rksep is sufficient for constructing representative sets.
\begin{dfn}\label{dfn:trimming}
Let $\P$ be a weighted family of \rsets and
let $\F$ be a family of \rksets.
The weighted family $\trim{\P}\subseteq \P$ is defined as follows:
For each $F\in \F$, and for each $1\leq i \leq k$, check if there exists some
$P\in \P$ with $|P|=i$ and $P\leq F$. If so, insert into $\trim{\P}$ some
$P\in \P$ that is of minimal weight among those with $|P|=i$ and $P\leq F$.
\end{dfn}

\begin{lemma}\label{lem:trim_is_rep}
Let $\F$ be an \rksep and let $\P$ be a weighted family of \rksets.
Then $\trim{\P}$ represents $\P$.
\end{lemma}
\begin{proof}
 Fix an \rkset $Q$ and suppose there is an \rkset $P\in \P$
 that is \rkcomp with $P$.
 In particular, we have $|P| = k-|Q|$.
 Since $\F$ is an \rksep, there exists $F\in \F$ such that $P\leq F \leq \comp{Q}$.
 As $P\leq F$, when constructing $\trim{\P}$
 we must have inserted into it an \rkset
 $P'\in \P$ of size $k-|Q|$ 
 such that $P'\leq F$ and $\wt(P')\leq \wt(P)$.
 Therefore $P'\leq \comp{Q}$, which implies that $P'+Q$
 is an \rkset. As $|P'+Q| = k$, we have that $P'$ and $Q$ are \rkcomp.
\end{proof}

We can now combine Lemma~\ref{lem:trim_is_rep} with the construction of an \rksep from~Corollary \ref{cor:multisetsepmain},
and thus obtain a construction of a small representative family for a weighted family of multisets.
\begin{cor}\label{cor:repset}
There exists a deterministic algorithm that, given a weighted family $\P$ or \rksets, 
runs in time $\O(|\P|\cdot \sepcompute)$ and returns
returns a family $\rep{P}\subseteq \P$ that represents $\P$ and has size $\repsize$.
\end{cor}
\begin{proof}
 Let $\F$ be the \rksep of size $\sepsize$ given by Corollary \ref{cor:multisetsepmain}; recall that
 $\F$ can be computed in time $\O(\sepcompute)$.
 We compute $\rep{P}=\trim{\P}$ which represents $\P$ by Lemma~\ref{lem:trim_is_rep}.
 The construction of $\trim{\P}$ amounts to going over all pairs of \rsets $P\in \P$ and $F\in \F$
 and checking whether $P\leq F$.
 Thus, the computation takes time at most $\O(|\P|\cdot |\F|\cdot n)  = \O(|\P|\cdot \sepcompute)$.

\end{proof}

\section{Construction of a separating family}\label{sec:sepfam}
The purpose of this section is to prove Theorem~\ref{thm:sepratingmain}, that is, to construct a small separating
family of functions. First, we need to introduce some auxiliary results.

\newcommand{\prodRect}{\prod_{i\in [t]} [m_i]}

\paragraph*{Hitting combinatorial rectangles.} We first recall the notion of a hitting set for combinatorial rectangles. For a sequence of integers $(m_1,m_2,\ldots,m_t)$, by $\prodRect$ we denote $[m_1]\times [m_2]\times \ldots \times [m_t]$.
\begin{dfn}\label{dfn:rechittingset}
Let $R\subseteq \prodRect$ be
a set of the form $R_1\times \ldots\times R_t$, where
$R_i\subseteq [m_i]$ for all $i\in [t]$.
We say that $R$ is a \emph{combinatorial rectangle with sidewise density $\gamma$},
if for every $i\in [t]$ we have that $|R_i|\geq \gamma \cdot m_i$.
A set $H\subseteq \prodRect$ is a \emph{hitting set for rectangles with sidewise density $\gamma$}
if for every set $R\subseteq \prodRect$
that is a combinatorial rectangle of sidewise density $\gamma$,
it holds that $R\cap H \neq \emptyset$.
\end{dfn}

Linial et al.~\cite{LLSZ97} gave the following construction of a hitting set for combinatorial rectangles.
\begin{thm}[\cite{LLSZ97}]\label{thm:LLSZhittingset}
A hitting set $H\subseteq [m]^t$ for rectangles with sidewise density $1/3$
of size $|H|= 2^{O(t)}\cdot m^{O(1)}$ can be constructed
in time $2^{O(t)}\cdot m^{O(1)}$.
\end{thm}

We need a hitting set for combinatorial rectangles in a universe where
the coordinates are from domains of different sizes.
We show that Theorem~\ref{thm:LLSZhittingset} can be adapted to this setting.
\begin{cor}\label{cor:LLSZhittingset}
Suppose $m_1,\ldots,m_t \leq m$.
Then a hitting set $H\subseteq \prodRect$ for rectangles with sidewise density $1/2$
of size $|H|= 2^{O(t)}\cdot m^{O(1)}$ can be constructed
in time $2^{O(t)}\cdot m^{O(1)}$.
\end{cor}
\begin{proof}
For the purpose of the proof it will be convenient to redefine $[a]$ as $\set{0,\ldots,a-1}$.
Let $m' = 3m$. Define mapping $\pi\colon [m']^t \to \prodRect$ as follows:
$$\pi(a_1,a_2,\ldots,a_t)=(a_1\bmod m_1,a_2\bmod m_2,\ldots,a_t\bmod m_t).$$
Observe that if $R_i\subseteq [m_i]$ is such that $|R_i|\geq m_i/2$, and $R_i'\subseteq [m']$ is the set of all elements of $[m']$ whose remainders modulo $m_i$ belong to $R_i$, then $|R_i'|\geq m'/3$. This follows from the fact that $m'=3m\geq 3m_i$. Therefore, if $R\subseteq \prodRect$ is a combinatorial rectangle with sidewise density $1/2$, then $R'\triangleq \pi^{-1}(R)\subseteq [m']^t$ is a combinatorial rectangle with sidewise density $1/3$. Let $H'\subseteq [m']^t$ be the hitting set for combinatorial rectangles with sidewise density $1/3$ given by Theorem~\ref{thm:LLSZhittingset}. Moreover, let $H\triangleq \pi(H')$; hence we have that $|H|\leq |H'|=2^{O(t)}\cdot m^{O(1)}$. We have that $H'\cap R'\neq \emptyset$, so also $H\cap R=\pi(H')\cap \pi(R')\supseteq \pi(H'\cap R')\neq \emptyset$. Since $R$ was chosen arbitrarily, we infer that $H$ is a hitting set for combinatorial rectangles with sidewise density~$1/2$.
\end{proof}




\paragraph*{Pairwise independent families.} Another component in our construction is a family of $\eps$-pairwise independent functions.
\begin{dfn}\label{dfn:almostpairwise}
A family of functions $\H\oftype{[n]}{[m]}$
is $\eps$-pairwise independent if for all $x,y\in [n]$ and $a,b\in [m]$, it holds that
\[|\pr_{f\leftarrow \H} (f(x) = a \wedge f(y) = b) -1/m^2| \leq \eps.\]
\end{dfn}
The constructions of~\cite{NN93} and~\cite{AGHP92}
imply the following construction of an $\eps$-pairwise independent family of functions
(cf. \cite[Section 4]{AYZ08}).

\begin{thm}[\cite{AGHP92,NN93}]\label{thm:almostpairwise}
Fix any $m\leq n$. Then a $1/m^2$-pairwise independent family $\H\oftype{[n]}{[m]}$
of size $m^{O(1)}\cdot \log n$ can be constructed in time $O(m^{O(1)}\cdot n\cdot \log n)$.
\end{thm}

We now use Theorem~\ref{thm:almostpairwise} to construct a small family that separates one element from $k$ other elements with non-negligible probability.

\begin{lemma}\label{lem:seperatingWithpairwise}
Fix integers $k,n$ with $1\leq k\leq n$. Then a family $\H_{k,4k}\oftype{[n]}{[4\cdot k]}$ that is $(1,k,4\cdot k)$-separating with
probability $1/2$ and has size $k^{O(1)}\cdot \log n$ can be constructed in time $k^{O(1)}\cdot n\cdot  \log n$.
\end{lemma}
\begin{proof}
Let $\H\oftype{[n]}{[4k]}$ be the family of $1/(4k)^2$-pairwise independent functions, given by Theorem \ref{thm:almostpairwise}. 
Fix sets $C=\set{a}$ and $D=\set{b_1,\ldots,b_{k-1}}$ we
want to separate. For $j\in [k-1]$, let $X_j$ be a random variable
that is equal to one if $f(a) = f(b_j)$ and to zero otherwise, where $f$ is chosen uniformly at random from $\H$.
From the $1/4k$-pairwise independence of $\H$, we have that $\E (X_j)  = \pr(X_j=1) \leq 4k\cdot (1/(4k)^2 + 1/(4k)^2) \leq 1/(2k)$.
Let us define $X = \sum_{j=1}^{k-1} X_j$, so that we have $\E(X) \leq (k-1)/(2k) < 1/2$.
Note that $X$ is precisely the number of collisions between $C$ and $D$.
From Markov's inequality, we have that $\pr(X\geq 1) \leq 1/2$.
So with probability at least $1/2$ over the choice of $f\in \H$, $C$ and $D$ are separated by $f$.
\end{proof}

The next claim follows from the well-known theorem that the geometric mean of non-negative
numbers is never larger than their arithmetic mean.
\begin{proposition}\label{prop:restprod}
Let $k_1,\ldots,k_t$ be non-negative real numbers such
that $k_1+\ldots+k_t \leq k$.
Then $\prod_{i=1}^t k_i  \leq (k/t)^t$.
\end{proposition}

\paragraph*{The main construction.} We are ready to proceed to the main construction.
\begin{proof}[Proof of Theorem~\ref{thm:sepratingmain}]
Fix disjoint subsets $C,D\subseteq [n]$ with $|C|=t$ and $|D|\leq k-t$.
Recall that we want to construct a family of functions $\H\oftype{[n]}{[t+1]}$,
such that for any $C$ and $D$ as above, there exists $h\in \H$
that separates $C$ from $D$.
It will be convenient to present the family by constructing $h$
adaptively given $C$ and $D$. That is, for arbitrarily chosen $C$ and $D$, we will adaptively construct a function $h$ that separates $C$ from $D$. Function $h$ will be constructed by taking a number of {\em{choices}}, where each choice is taken among a number of possibilities. The final family $\H$ will comprise all $h$ that can be obtained using any such sequence of choices; thus, the product of the numbers of possibilities will limit the size of $\H$. As $C$ and $D$ are taken arbitrarily, it immediately follows that such $\H$ separates every pair $(C,D)$. 

\begin{enumerate}
\item Let $\H_0\oftype{[n]}{[k^2]}$ be the $k$-perfect hash family given by
Theorem~\ref{thm:hash_to_k^2}. Choose $f_0 \in \H_0$ that is injective on $C\cup D$ --- there are $k^{O(1)}\cdot \log n$ choices for this stage.

From now on, we identify $C$ and $D$ with their images in $[k^2]$ under $f_0$.
\item Let $\H_1\oftype{[k^2]}{[t]}$ be the $t$-perfect hash family given by Theorem~\ref{thm:NSShash}. Choose a function $f_1\in \H_1$ that is injective on $C$ --- there are $e^{t+ O(\log^2 t)}\cdot \log k$ choices for this stage.

For $i\in [t]$, we denote from now on by $c_i$ the (unique) element of $C$ with $f_1(c_i) = i$. Moreover, elements mapped to $i$ under $f_1$ will be denoted by $\B_i$, and we will think of them as the {\em{$i$-th bucket}}.
\item Choose non-negative integers $k_1,\ldots, k_t$ such that $k_i$ is the number of elements
$a\in D$ with $f_1(a) = i$. Note that $k_1+\ldots + k_t \leq k-t$,
 so the number of choices for this stage is at most $\binom{k}{t} \leq (ek/t)^t$.
\item For $i\in [t]$, let $\H_{k_i,4k_i}\oftype{[k^2]}{[4\cdot k_i]}$ be family given by Lemma~\ref{lem:seperatingWithpairwise}. That is, $\H_{k_i,4k_i}$ is
$(1,k_i,4\cdot k_i)$-separating with
probability $1/2$ and has size $m_i=k_i^{O(1)}\cdot \log k$.

Let $R$ be the set of all vectors $(h_1,\ldots,h_t)$ such that
for all $i\in [t]$, $h_i$ is an element of $\H_{k_i,4k_i}$
that separates $c_i$ from $D_i$.
Identify $[m_i]$ with $\H_{k_i,4k_i}$ by ordering the functions in $\H_{k_i,4k_i}$ arbitrarily.
Observe that $R$ is a combinatorial rectangle of sidewise density $1/2$
in $\prodRect$. Note that for all $i\in [t]$, $m_i\leq m$ for some $m=k^{O(1)}$.
Let $H$ be the hitting set for combinatorial rectangles with sidewise density
$1/2$ given by Corollary~\ref{cor:LLSZhittingset}.
Choose an element $(h_1,\ldots,h_t) \in H\cap R$. As $|H|=2^{O(t)}\cdot k^{O(1)}$, there are $2^{O(t)}\cdot k^{O(1)}$ choices for this stage.

For $i\in [t]$, apply $h_i$ to partition bucket $\B_i$ into buckets $\B_{i,1},\ldots,\B_{i,4k_i}$; that is, element $a\in [n]$ is put into bucket $\B_{f_1(a),h_{f_1(a)}(a)}$.
Since $(h_1,\ldots,h_t) \in R$, we thus have that each $c_i$ is mapped to a separate bucket from all elements of $D_i$.
\item For each $i\in [t]$, choose the unique index $j_i\in [4k_i]$ such
that $c_i$ was mapped to bucket $\B_{i,j_i}$.
Construct the final function $h$ by mapping all the elements of $\B_{i,j_i}$ to $i$, and mapping all the elements of
$\B_{i,j'}$ for $j'\neq j_i$ to $t+1$, for all $i\in [t]$. The number of choices for this stage is
$4^t\cdot \prod_{i=1}^t k_i \leq (4k/t)^t$,
where the inequality follows from Proposition~\ref{prop:restprod}.
\end{enumerate}

To sum up, the number of different functions enumerated in all the above stages is at most
\[k^{O(1)}\cdot \log n\cdot e^{t+ O(\log^2 t)}\cdot \log k\cdot (ek/t)^t\cdot 2^{O(t)}\cdot k^{O(1)}\cdot (4k/t)^t
=(k/t)^{2t}\cdot 2^{O(t)}\cdot k^{O(1)} \cdot \log n\]
\[ =\O((k/t)^{2t}\cdot 2^{O(t)}\cdot \log n).\]
Similarly, going over the construction times of the different components used in the above stages,
we get that the family can be constructed in time
\[O(k^{O(1)}\cdot n\cdot \log n\cdot e^{t+ O(\log^2 t)}\cdot k\cdot \log k\cdot (ek/t)^t\cdot 2^{O(t)}\cdot k^{O(1)}\cdot (4k/t)^t)\]
\[ = O((k/t)^{2t}\cdot 2^{O(t)}\cdot k^{O(1)} \cdot n\cdot \log n)=\O((k/t)^{2t}\cdot 2^{O(t)}\cdot n\cdot \log n).\]\end{proof}

\paragraph*{Lopsided universal sets.} 
Informally speaking, an \emph{\lops{p}{q}} is a set of strings such that, when focusing on any $k\triangleq p+q$ locations, we see all patterns of hamming weight $p$. 
%
%
%
Universal set families, and in particular the lopsided ones, have important applications in the determinization of parameterized algorithms that use the technique of color coding; cf.~\cite{FLS14}.
\begin{dfn}\label{dfn:lopsided}
An \lops{p}{q} is a family of subsets $\F \subseteq [n]$ such that for every disjoint subsets $A,B \subseteq [n]$ with $|A|=p$ and $|B|=q$ there exists $F\in \F$ with $A\subseteq F$ and $F\cap B = \emptyset$.
\end{dfn}

A probabilistic argument shows the existence of \lopss{p}{q} of size $\binom{k}{p}\cdot \sqrt{p\cdot q \cdot k} \cdot \log n$ (cf. Lemma 52 in \cite{Bshouty14}). Fomin, Lokshtanov and Saurabh \cite{FLS14} used the technique of \cite{NSS95} to construct an \lops{p}{q} of size $\binom{k}{p}\cdot 2^{O\left(\frac{k}{\log\log(k)}\right)}\cdot \log n.$
Consider the case $p=o(q)$: For simplicity of discussion, let us focus on the  $\binom{k}{p}\leq (e\cdot k/p)^p$ term in the probablistic construction and omit $\log n$ terms. In this case the $2^{O\left(\frac{k}{\log\log(k)}\right)}$ term of $\cite{FLS14}$ is much larger than $\binom{k}{p}$. Bshouty \cite{Bshouty14} gives a better construction for such $p$ achieving size roughly $k^{p+2}$ (see Lemma 53 in \cite{Bshouty14}). We remark that Bshouty~\cite{Bshouty14} referred to this object as a cover free family. As a corollary of the last section we obtain an explicit construction of size $(k/p)^{O(p)}$. 
%
%
\begin{thm}\label{thm:lopsided_main}
Fix integers $p,q$ and $n$ such that $k=p+q \leq n$.
Then an \lops{p}{q} of size $ \O((k/p)^{2\cdot p}\cdot 2^{O(p)}\cdot \log n)$
can be constructed in time $ \O((k/p)^{2\cdot p}\cdot 2^{O(p)}\cdot n\cdot \log n)$.
\end{thm}
\begin{proof}
Let $\H\oftype{[n]}{[p+1]}$ be the $(p,k)$-minimal separating family of size $\O((k/p)^{2p}\cdot 2^{O(p)}\cdot \log n)$
given by Theorem~\ref{thm:sepratingmain}.
For each $h\in \H$, construct set $F_h=h^{-1}([p])$, and let $\F\triangleq \set{F_h\mid h\in \H}$.
The definition of a minimal separating family immediately implies that $\F$ is an \lops{p}{q}, and the time needed for the construction follows from Theorem~\ref{thm:sepratingmain}.
\end{proof}


\section{Algorithmic applications}\label{sec:algapp}

In this section we present applications of our construction of representative sets for multisets to the design of deterministic algorithms for problems with relaxed disjointness constraints. We first give some useful auxiliary tools, and then present algorithms for three problems: \rsimp, \rpacking, \rmondetect. Let us remark that the algorithms for the first two problems follow from the algorithm for the (more general) last problem. Nevertheless, we find it more didactic to present the applications in increasing order of difficulty.
Throughout this section, $r$ always denotes an integer in the range $1<r<k$.

\subsection{Algorithmic preliminaries}

The following simple observation will be used in all of
our algorithmic applications.
\begin{lemma}\label{lem:repsetnotempty}
Let $\P$ be a family of \rksets that constains a multiset $P$ of size $k$
and weight $w$.
Then any $\rep{\P}\subseteq \P$ that represents $\P$ contains a
multiset $P'$ of size $k$ and weight $w'\leq w$.
\end{lemma}
\begin{proof}
 Take $B=\emptyset$.
 Then $P$ and $B$ are \rkcomp.
 So there must exist $P'\in \rep{\P}$ that
 is \rkcomp with $B$ and has weight at most $w$.
 As $|P'+B|=k$, we have $|P'|=k$.
\end{proof}

The following simple fact is immediate and will be used implicitly.
\begin{proposition}
If $\A''$ represents $\A'$ and $\A'$ represents $\A$, then $\A''$ represents $\A$.
\end{proposition}
We now show that representative sets behave robustly under union of families.
\begin{lemma}\label{lem:repsetsunion}
Suppose $\A$ and $\B$ are two weighted families of $\rksets$, and suppose further that $\rep{\A}$ represents $\A$ and $\rep{\B}$ represents $\B$.
Then $\rep{\A}\cup \rep{\B}$ represents $\A\cup \B$.
\end{lemma}
\begin{proof}
Take any \rkset $D$ such that there exists $C \in \A\cup \B$ that is \rkcomp with $D$. If $C\in \A$, then there exists $C'\in \rep{\A}$ with $\wt(C')\leq \wt(C)$ such that $C'$ is \rkcomp with $D$. Then also $C'\in \rep{A}\cup \rep{B}$ and we are done. The reasoning for the case when $C\in \B$ is symmetric.
\end{proof}
Finally, we introduce operation $\bullet$ that extends the considered weighted family of \rksets. This is the crucial ingredient in the classic technique of iteratively extending the current family by a new object, and then shrinking the size of the family by taking a representative subfamily.
\begin{dfn}
We say that \rksets $A$ and $B$ are \emph{\rkcons},
if $|A+B|\leq k$ and $A_i+B_i\leq r$ for all $i\in [n]$.
(Note that this is a weaker definition than being \rkcomp, where it is required that $|A+B|=k$.)

\noindent Let $\A$ and $\B$ be weighted families of \rksets. We define the weighted family
$\A\bullet \B$ as follows:
$$\A\bullet \B \triangleq \set{A+B \mid A\in \A, B\in \B, A \textrm{ and } B \textrm{ are \rkcons}}.$$
In particular,
for an element $i\in [n]$,
we denote 
$$\A\bullet i\triangleq \set{A+\set{i}\mid A\in \A, |A|<k, A_i<r}.$$
\end{dfn}
\begin{lemma}\label{lem:bulletrep}
Suppose $\A$ and $\B$ are two weighted families of $\rksets$, and suppose further that $\rep{\A}$ represents $\A$ and $\rep{\B}$ represents $\B$.
Then $\rep{\A}\bullet \rep{\B}$ represents $\A\bullet \B$.
\end{lemma}
\begin{proof}
Fix an \rkset $D$ such that
there exists $C \in \A\bullet \B$
that is \rkcomp with $D$.
Hence, we have that $C =A + B$ for
some $A\in \A$ and $B\in \B$ that are \rkcons.
In particular, $A$ is \rkcomp with $B+ D$.
Hence, there exists $A' \in \rep{A}$ with $\wt(A')\leq \wt(A)$ that is \rkcomp with $B+D$.
In other words, $A'+B+D$ is an \rkset of size $|A'+B+D|=k$.
It follows that $B$ is \rkcomp with $A'+D$.
Hence, there exists $B'\in \rep{B}$ with $\wt(B')\leq \wt(B)$ that is \rkcomp with $A+D$.
Again, for this $B'$ we have that $A'+B'+D$ is an \rkset of size
$k$.
Therefore, $A'+B'$ is an \rkset of size at most $k$ and weight not larger than the weight of $A+B$.
In other words, $A'$ and $B'$ are \rkcons, and therefore $A'+B'\in \A'\bullet \B'$.
As $A'+B'$ is \rkcomp with $D$ and $\wt(A'+B')\leq \wt(A+B)$, we are done.
\end{proof}

\subsection{\rsimp}

We first introduce some notation for defining the \rsimp problem.
Let $G$ be a directed graph on $n$ vertices.
A \emph{$k$-path} in $G$ is a \emph{simple} walk
$P=v_1\to v_2\to \ldots \to v_k$ in $G$, that is, all vertices traversed by the walk are pairwise different
An \emph{\rkpath} is a walk
$P=v_1\to v_2\to \ldots \to v_k$ in $G$
such that no vertex is repeated more than $r$ times. In case the vertex set of $G$ is equipped with a weight function, the weight of an \rkpath $P$ is defined as the weight of the multiset of vertices traversed by $P$, where each vertex is taken the number of times equal to the number of its visits on $P$. 

\defparproblemu{\rsimp}{Directed graph $G=(V,E)$, weight function $\wt\colon V\to {\mathbb R}$, integers $r$ and $k$.}{$r,k$}{Determine whether there exists an \rkpath in $G$, and if so then return one of minimal possible weight.}

%
%
%

\begin{thm}\label{thm:rsimpalg}
\emph{\rsimp} can be solved in deterministic time
$\O(\repcomputenon \cdot n^3\cdot \log n)$.
\end{thm}

\begin{proof}
Identify the vertex set $V$ with $[n]$. For $u\in [n]$ and $i\in [k]$, let
$\P_{i,u}$ denote the set of \rpaths in $G$ that have length $i$
and end in $u$.
Note that an element $P$ of such a set $\P_{i,u}$
can be viewed as an \rkset when we ignore the
order of vertices in the path.
We can thus view the sets $\P_{i,u}$ as
families of \rksets.
We will compute, using a dynamic programming algorithm, for each $u\in [n]$ and $i\in [k]$,
a set $\rep{\P}_{i,u}\subseteq \P_{i,u}$ that represents
$\P_{i,u}$ and has size at most $\repsize$. By Lemma~\ref{lem:repsetnotempty}, a minimal-weight \rkpath in $G$ can be recovered by taking a minimum-weight multiset among families $\rep{\P}_{k,u}$, for $u\in [n]$. In particular, if all these sets are empty, then no \rkpath exists in $G$. 
Thus, such an algorithm can be used to solve \rsimp.

We proceed with the algorithm description.
For $i=1$ and $u\in [n]$, family $\rep{\P}_{1,u}$ simply contains only the length-one path $u$.
Assume that we have computed families $\rep{\P}_{i,u}$ for every $u\in [n]$.
We describe the computation of $\rep{\P}_{i+1,v}$ for a fixed $v\in [n]$,
together with its running time.

\begin{enumerate}
\item We compute $\P'\triangleq  \bigcup_{(u,v)\in E} \rep{\P}_{i,u} \bullet v$.
Clearly, we have that $|\P'| \leq n\cdot \repsize $.
We now claim that $\P'$ represents $\P_{i+1,v}$.
Note first that $\P_{i+1,v} = \bigcup_{(u,v)\in E} \P_{i,u} \bullet v$.
By Lemma~\ref{lem:bulletrep}, for any $u\in [n]$, $\rep{\P}_{i,u}\bullet v$
represents $\P_{i,u} \bullet v$.
Therefore, $\P' = \bigcup_{(u,v)\in E} \rep{\P}_{i,u} \bullet v$ represents
$\bigcup_{(u,v)\in E} \P_{i,u} \bullet v = \P_{i+1,v}$.

Computing $\P'$ directly from the definition requires time $\O(\sum _{u\in [n]} |\rep{\P}_{i,u}|) =
\O(n \cdot \repsiz)$.

\item Now, using the algorithm of Corollary~\ref{cor:repset}, compute a set $\rep{\P}_{i+1,v}$ that represents $\P'$ and has size $|\rep{\P}_{i+1,v}| = \repsize$.
    This takes time $\O(|\P'|\cdot\sepcompute)= \O(n^2 \cdot \repcomputenopoly)$.
\end{enumerate}

\noindent Since during each of $k$ steps of the algorithm, this procedure is applied to every vertex $v\in [n]$, we conclude that the running time of the algorithm is $\O(\repcomputenon \cdot n^3\cdot \log n)$, as claimed.
\end{proof}

We have considered the {\sc $r$-Simple $k$-Path} problem where the weights are on the vertices. One could also consider the {\em edge weighted} variant where every edge has a weight and the weight of a walk is the sum of the weight of the edges on the walk (counting multiplicities of edges). One can reduce this variant to the vertex weighted variant as follows. Let $G$ be the input graph, we make a new graph $G'$ from $G$. Here each edge $uv$ of $G$ is replaced by a new vertex $x_{uv}$ with $u$ being the only in-neighbor and $v$ being the only out-neighbor of $x_{uv}$. The weight of $x_{uv}$ is set to the weight of the edge $uv$. The weight of the vertices of $G'$ that correspond to the vertices of $G$ is set to $0$. An $r$-simple $k$-path in $G$ corresponds to an $r$-simple $(2k+1)$-path in $G'$ of the same weight. On the other hand, a minimum weight  $r$-simple $(2k+1)$-path in $G'$ must start and end in vertices corresponding to vertices of $G$, since these have weight $0$.  Such $r$-simple $(2k+1)$-paths in $G'$ correspond to an $r$-simple $k$-path in $G$  of the same weight. Thus we can run the algorithm for Theorem~\ref{thm:rsimpalg} on $G'$, obtaining an algorithm for the edge weighted variant with running time $r^{O(k/r)}n^{O(1)}$.

\subsection{\rpacking}

We say that $P\subseteq [n]$ is a \qset if $|P|=q$.
We say that a family of \qsets $\A=\set{P_1,\ldots, P_t}$
is an \emph{\rpack} if every element $j\in [n]$ appears in at most $r$
of the \qsets in $\A$.
The \rpacking problem asks whether a given family of \qsets contains
an \rpack. Again, in case universe $[n]$ is equipped with a weight function, we can say that a family of \qsets is {\em{weighted}} by setting the weight of a \qset to be the total weight of its elements. The weight of an \rpack is defined as the sum of weights of its sets.


\defparproblemu{\rpacking}{Weighted family $\F$ of \qsets, integers $r,p,q$}{$r,p,q$}{Determine whether there exists a subfamily $\A\subseteq \F$
with $|\A| = p$ that is an \rpack, and if so then return such $\A$ of minimal total weight.}

%
%
%

\begin{thm}\label{thm:rpack}
Let $k\triangleq p\cdot q$.
Then \rpacking can be solved in deterministic time $\O(|\F|\cdot \repcomputenon \cdot n\log^2 n)$.
\end{thm}
\begin{proof}
Let $\F$ be the input family of \qsets.
For $0\leq i\leq p$, denote by $\P_i$ the set of \rpacks of size $i$ in $\F$.
Note that an element $\A= \set{P_1,\ldots,P_i}$ of $\P_{i}$
can be viewed as an \rkset: We think of $\A$ as the multiset $\A=P_1+\ldots+P_i$.
We know that this multiset has size at most $p\cdot q= k$ and every element $j\in [n]$ appears
in $\A$ at most $r$ times. We will compute, for each $0\leq i\leq p$,
a subfamily $\rep{\P}_{i}\subseteq \P_{i}$ that represents
$\P_{i}$ and has size at most $\repsize$. By Lemma~\ref{lem:repsetnotempty}, a minimum-weight \rpack in $\F$ can be recovered by taking a minimum-weight element of $\rep{\P}_{p}$. In particular, if $\rep{\P}_p$ is empty, then no \rpack exists in $\F$. Thus, such an algorithm can be used to solve \rpacking. 

We proceed with the algorithm description. For $i=0$, we take $\rep{\P}_0 = \P_0 = \set{\emptyset}$. Assume we have computed $\rep{\P}_i$ for some $0\leq i<p$.
We describe the computation of $\rep{\P}_{i+1}$.

\begin{enumerate}
\item We compute $\P'\triangleq  \rep{\P}_{i} \bullet \F$.
Clearly, we have that $|\P'| \leq |\F|\cdot |\rep{\P}_i|\leq |\F| \cdot \repsize$.
We claim that $\P'$ represents $\P_{i+1}$.
Note first that $\P_{i+1} = \P_i \bullet \F$.
Hence, by Lemma~\ref{lem:bulletrep}, $\P' = \rep{\P}_{i}\bullet \F$
represents $\P_{i} \bullet \F= \P_{i+1}$.

Computing $\P'$ directly from the definition requires time $\O(|\F| \cdot |\rep{P}_i|\cdot n)=\O(|\F|\cdot \sepcompute)$.

\item Now, using the algorithm of Corollary~\ref{cor:repset}, compute a set $\rep{\P}_{i+1}$ that represents $\P'$ and has size $|\rep{\P}_{i+1}| =\repsize$.
    This takes time $\O(|\P'|\cdot \sepcompute)= \O(|\F| \cdot \repcomputenon\cdot  n\log^2 n)$.
\end{enumerate}
Since we repeat this procedure $p$ times, the claimed running time follows
\end{proof}

In Theorem~\ref{thm:rpack} we gave an algorithm for \rpacking where every element has a weight and the weight of a set is equal to the sum of the weights of the elements. A related variant is one where a weight function $w : {\cal F} \rightarrow \mathbb{N}$ is given as input. That is, every set $P \in {\cal F}$ has its own weight $w(P)$, and the weight of a subfamily ${\cal A} \subseteq {\cal F}$ is the sum of the weights of the sets in ${\cal A}$. We will call this the {\em set weighted} \rpacking problem. The set weighted problem can be reduced to the original by adding for every set $P \in {\cal F}$ a new element $e_P$ of weight $w(P)$, inserting $e_P$ into $P$ and giving $e_P$ weight $w(P)$. Note that $P$ is the only set in the new instance that contains $e_P$. All elements corresponding to the elements of the instance of set weighted \rpacking are given weight $0$. All sets in the new instance have size $q+1$, and $r$-packings in the new instance correspond to $r$-packings in the original instance with the same weight. Thus we can apply the algorithm of Theorem~\ref{thm:rpack} on the new instance in order to solve the instance of set weighted \rpacking. This yields an algorithm for set weighted \rpacking with running time $|{\cal F}|r^{O(k/r)}n^{O(1)}$.

\subsection{\rmondetect}
In this subsection we consider polynomials from ring $\Z[X_1,\ldots,X_n]$. We say a monomial $X_1^{d_1}\cdots X_n^{d_n}$
is an \emph{\rmon} if for all $i\in [n]$, we have $d_i\leq r$.
We say the monomial is an \emph{\rkmon} if the above holds
and the total degree of the monomial is $k$.
Let $C$ be an arithmetic circuit computing a polynomial $f(X_1,\ldots,X_n)\in \Z[X_1,\ldots,X_n]$.
We say $C$ is \emph{non-canceling} if it contains only variables at its leaves (i.e., no constants),
and only addition and multiplication gates (i.e., no substractions).
For a non-canceling circuit $C$, we define $|C|$ to be the number of multiplication and addition gates plus
the number of leaves (each containing a variable). We assume the fan-in of a non-canceling circuit is two, i.e.,
each multiplication and addition gate has at most two wires coming in. (This will simply be convenient for bounding the running time as a function of $|C|$.)

\defparproblemu{\rmondetect}{A non-canceling circuit $C$ computing a polynomial $f(X_1,\ldots,X_n)\in \Z[X_1,\ldots,X_n]$, integers $r,k$.}{$r,k$}{Determine if there exists an \rkmon in $f$ with non-zero coefficient and if so, then return such a monomial.}

%
%
%

\begin{thm}\label{thm:rmonomial}
Given a non-canceling circuit $C$ computing a polynomial $f(X_1,\ldots,X_n)\in \Z[X_1,\ldots,X_n]$,
 \rmondetect can be solved in deterministic time $\O(|C|\cdot \mondetectcomputenon\cdot n\log^3 n)$.
\end{thm}
\begin{proof}
 For each gate $s$ of $C$, let $f_s$ be the
 polynomial computed at $s$. Define $\P_s$ to be the set of \rmons of total degree at most $k$ that appear in
 $f_s$ with nonzero coefficient.
 We can view $\P_s$ as a family of \rksets: An \rmon $M = X_1^{d_1}\cdots X_n^{d_n}$ of total degree at most $k$
 corresponds to the \rkset where each $j\in [n]$ appears $d_j$ times (in this case, we put uniform weights on the elements of universe $[n]$).
 We present an algorithm that computes, for every gate $s$ of $C$,
 a subfamily $\rep{\P}_s\subseteq \P_s$ of size $\repsize$ that represents $\P_s$.
 Let $s_{out}$ be the output gate of $C$.
 By Lemma~\ref{lem:repsetnotempty}, if $f$ contains an \rmon of total degree $k$, then
 $\rep{\P}_{s_{out}}$ will contain such a monomial. Hence, to solve \rmondetect it suffices to check whether $\rep{\P}_{s_{out}}$ is nonempty.

 We compute the sets $\rep{\P}_{s}$ from bottom to top. For $s$ being an input gate containing variable $X_i$, we simply put $\rep{\P}_s=\P_s=\set{\set{i}}$.
  Take then any non-input gate $s$, and suppose we have computed $\rep{\P}_{s_1}$ and $\rep{\P}_{s_2}$
 for the gates $s_1$ and $s_2$ having wires into $s$.
 \begin{enumerate}
  \item If $s$ is an addition gate, then we define $\P'\triangleq\rep{\P}_{s_1}\cup\rep{\P}_{s_2}$.
  We claim that $\P'$ represents $\P_s$:
 Since $C$ is non-canceling, the set of monomials that appear in $f_s$ with a nonzero coefficient
 is simply the union of the sets of monomials of appearing in $f_{s_1}$ and in $f_{s_2}$.
 In particular, $\P_s = \P_{s_1}\cup \P_{s_2}$.
 Therefore, by Lemma~\ref{lem:repsetsunion} we infer that $\P' = \rep{\P}_{s_1}\cup \rep{\P}_{s_2}$ represents $\P_s$. Note that $|\P'|\leq |\rep{\P}_{s_1}|+|\rep{\P}_{s_2}|=\repsize$ and $\P'$ can be computed in time $\O(|\rep{\P}_{s_1}|\cdot|\rep{\P}_{s_2}| \cdot n)=\O(\repcomputenon \cdot n\log^2 n)$.
  \item If $s$ is a multiplication gate, then we define $\P' \triangleq \rep{\P}_{s_1}\bullet \rep{\P}_{s_2}$.
  Since $C$ is non-canceling, the set of monomials appearing in $f_s$ is exactly the set of all products of a monomial
  appearing in $f_{s_1}$ and a monomial appearing in $f_{s_2}$. In particular, $\P_s$ is exactly the set of these products that are also
  \rkmons.
  When viewed as a multiset, the product of monomials $M_1$ and $M_2$
  is the multiset $M_1+M_2$.
  Thus, we have that $\P_s = \P_{s_1}\bullet \P_{s_2}$ and therefore, using Lemma \ref{lem:bulletrep}, we infer that
  $\P' = \rep{\P}_{s_1} \bullet \rep{\P}_{s_2}$ represents $\P_s$. Note that $|\P'|\leq |\rep{\P}_{s_1}|\cdot|\rep{\P}_{s_2}|=\O(\repcomputenon \cdot \log^2 n)$ and $\P'$ can be computed in time $\O(|\rep{\P}_{s_1}|\cdot|\rep{\P}_{s_2}| \cdot n)=\O(\repcomputenon \cdot n\log^2 n)$.

   \item  Now, using the algorithm of Corollary~\ref{cor:repset}, compute subfamily $\rep{\P}_{s}$ of size $\repsize$ that represents $\P'$.
    This takes time $\O(|\P'|\cdot \sepcompute)= \O(\mondetectcomputenon\cdot n\log^3 n)$.
   \end{enumerate}
Since the above procedure is applied to every gate of $C$, the claimed running time follows.
\end{proof}

We remark that, after a trivial modification, the algorithm above can equally easily solve also a weighted variant of \rmondetect, where each variable is equipped with a weight and we are interested in extracting a monomial of minimum total weight, defined as the sum of the weights of variables times their degrees. It is not hard to reduce the problems \rsimp and \rpacking, considered in the previous sections, to this variant; we leave the details to the reader.

\newcommand{\probSpanTree}{{\sc{Degree-Bounded Spanning Tree}}\xspace}
\newcommand{\Oh}{O}
\newcommand{\Ohstar}{\Oh^*}

\section{Low degree monomials and low degree spanning trees}\label{sec:kirchoff}

Let $G$ be a simple, undirected graph with $n$ vertices. Let $\mathcal{T}$ be the family of spanning trees of $G$; In particular, if $G$ is not connected then $\mathcal{T}=\emptyset$. With every edge $e\in E(G)$ we associate a variable $y_e$. The {\em{Kirchhoff's polynomial}} of $G$ is defined as:
$$K_G((y_e)_{e\in E(G)})=\sum_{T\in \mathcal{T}} \prod_{e\in E(T)} y_e.$$
Thus, $K_G$ is a polynomial in $\mathbb{Z}[(y_e)_{e\in E(G)}]$. Let $v_1,v_2,\ldots,v_n$ be an arbitrary ordering of $V(G)$. The {\em{Laplacian}} of $G$ is an $n\times n$ matrix $L_G=[a_{ij}]$, where
$$
a_{ij}=\begin{cases} \sum_{e\textrm{ incident to }v_i}\, y_e & \mbox{if $i=j$,}\\ -y_{v_iv_j} & \mbox{if $i\neq j$ and $v_iv_j\in E(G)$,}\\ 0 & \mbox{if $i\neq j$ and $v_iv_j\notin E(G)$.}\end{cases}
$$
Observe that $L_G$ is symmetric and the entries in every column and in every row of $L_G$ sum up to zero. Then it can be shown that all the {\em{first cofactors}} of $L_G$, i.e., the determinants of matrix $L_G$ after removing a row and a column with the same indices, are equal. Let $N_G$ be this common value; then $N_G$ is again a polynomial over variables $(y_e)_{e\in E(G)}$. The Kirchhoff's Matrix Tree Theorem, in its general form, states that these two polynomials coincide.

\begin{thm}[Matrix Tree Theorem]\label{thm:kirchhoff}
$K_G=N_G$.
\end{thm}

We remark that the Matrix Tree Theorem is usually given in the more specific variants, where all variables $y_e$ are replaced with $1$; then the theorem expresses the number of spanning trees of $G$ in terms of the first cofactors of $L_G$. However, the proof can be easily extended to the above, more general form; see for instance~\cite{tutte-book}.

Observe that Theorem~\ref{thm:kirchhoff} provides a polynomial-time algorithm for evaluating $K_G$ over a given field and vector of values of variables $(y_e)_{e\in E(G)}$. Indeed, we just need to construct matrix $L_G$, remove, say, the first row and the first column, and compute the determinant. We now present how this observation can be used to design a fast exact algorithm for the \probSpanTree problem.

\begin{thm}\label{thm:positive-span}
The \probSpanTree problem can be solved in randomized time $\Ohstar(d^{\Oh(n/d)})$ with false negatives.
\end{thm}
\begin{proof}
Associate every vertex $v$ of the given graph $G$ with a distinct variable $x_v$. Let $\overline{K}_G\in \mathbb{Z}[(x_v)_{v\in V(G)}]$ be a polynomial defined as $K_G$ with every variable $y_{uv}$, for $uv\in E(G)$, evaluated to $x_ux_v$. Then it follows that
$$\overline{K}_G((x_v)_{v\in V(G)})=\sum_{T\in \mathcal{T}} \prod_{v\in V(G)} x_v^{\deg_T(v)}.$$
 Observe also that $\overline{K}_G$ is $2(n-1)$-homogenous, that is, all the monomials of $\overline{K}_G$ have their total degrees equal to $2(n-1)$. Thus, graph $G$ admits a spanning tree with maximum degree at most $d$ if and only if polynomial $\overline{K}_G$ contains a $(d,2(n-1))$-monomial. Using Theorem~\ref{thm:kirchhoff} we can construct a $\poly$-sized circuit evaluating $\overline{K}_G$. Hence, verifying whether $\overline{K}_G$ contains a $(d,2(n-1))$-monomial boils down to applying the algorithm of Abasi et al.~\cite{ABGH14} for \rmondetect with $r=d$ and $k=2(n-1)$. This algorithm runs in randomized time $\Ohstar(d^{\Oh(n/d)})$ and can only produce false negatives.
\end{proof}

Let us repeat that in the proof of Theorem~\ref{thm:positive-span} we could not have used Theorem~\ref{thm:rmonomial} instead of the result of Abasi et al.~\cite{ABGH14}, because the constructed circuit is not non-cancelling. Derandomizing the algorithm and extending it to the weighted setting remains hence open.

Interestingly, the running time of the algorithm of Theorem~\ref{thm:positive-span} is essentially optimal, up to the $\log d$ factor in the exponent. A similar lower bound for \rsimp was given by Abasi et al.~\cite{ABGH14}.

\begin{thm}\label{thm:negative-span}
Unless ETH fails, there exists a constant $s>0$ such that for no fixed integer $d\geq 2$ the \probSpanTree problem with the degree bound $d$ can be solved in time $\Ohstar(2^{sn/d})$.
\end{thm}
\begin{proof}
It is known (see e.g.~\cite{the-awesome-book}) that, assuming ETH, there exists a constant $s>0$ such that the {\sc{Hamiltonian Path}} problem cannot be solved in time $\Ohstar(2^{sn})$ on $n$-vertex graphs. Consider the following reduction from {\sc{Hamiltonian Path}} to \probSpanTree with the degree bound $d$: given an instance $G$ of {\sc{Hamiltonian Path}}, create $G'$ by attaching to every vertex $v\in V(G)$ a set of $d-2$ degree-$1$ vertices, adjacent only to $v$. Since the new vertices have to be leaves in every spanning tree of $G'$, it follows that every spanning tree $T'$ of $G'$ is in fact a spanning tree of $G$ with all the vertices of $V(G')\setminus V(G)$ attached as leaves. In particular, $G'$ admits a spanning tree with maximum degree $d$ if and only if $G$ admits a spanning tree with maximum degree $2$, i.e., a hamiltonian path.

Observe that $|V(G')|=(d-1)\cdot |V(G)|$. Hence, if there was an algorithm solving \probSpanTree with the degree bound $d$ in time $\Ohstar(2^{sn/d})$, then by composing the reduction with the algorithm we would be able to solve {\sc{Hamiltonian Path}} on an $n$-vertex graph in time $\Ohstar(2^{s(d-1)n/d})\leq \Ohstar(2^{sn})$, thus contradicting ETH.
\end{proof}

\section{Conclusions}\label{sec:concl}
In this paper we considered relaxation parameter variants of several well studied problems in parameterized complexity and exact algorithms. We proved, somewhat surprisingly, that instances with moderate values of the relaxation parameter are significantly easier than instances of the original problems. We hope that our work, together with the result of Abasi et al.~\cite{ABGH14} breaks the ground for a systematic investigation of relaxation parameters in parameterized complexity and exact algorithms. We conclude with mentioning some of the most natural concrete follow up questions to our work.

\begin{itemize}\setlength\itemsep{-.7pt}
\item We gave a determinstic algorithm for {\em non-cancelling}  $(r,k)$-{\sc Monomial Detection} with running time $2^{O(k\frac{\log r}{r})}|C|n^{O(1)}$, while Abasi et al.~\cite{ABGH14} gave a randomized algorithm with such a running time for $(r,k)$-{\sc Monomial Detection} without the non-cancellation restriction. Is there a determinstic algorithm for $(r,k)$-{\sc Monomial Detection} with running time $2^{O(k\frac{\log r}{r})}|C|n^{O(1)}$?

\item Does there exist a deterministic algorithm with running time $2^{O(n\frac{\log r}{r})}$ for {\sc Degree-Bounded Spanning Tree}? Note that a deterministic algorithm with running time $2^{O(k\frac{\log r}{r})}|C|n^{O(1)}$ for $(r,k)$-{\sc Monomial Detection} immediately imply such an algorithm for {\sc Degree-Bounded Spanning Tree}. 

\item Is there a $2^{O(k\frac{\log r}{r})}n^{O(1)}$ time algorithm for the problem where we are given as input a graph $G$, integers $k$ and $d$, and asked whether $G$ contains a subtree $T$ on at least $k$ vertices, such that the maximum degree of $T$ is at most $d$? Observe that for $k=n$ this is exactly the {\sc Degree-Bounded Spanning Tree} problem.

\item Is it possible to obtain polynomial kernels for problems with relaxation parameters with smaller and smaller size bounds as the relaxation parameter increases?

\item Are the $\log r$ (or $\log d$) factors in the exponents of our running time bounds necessary? For example, is there an algorithm for {\sc $r$-Simple $k$-Path} with running time $2^{O(k/r)}n^{O(1)}$? Or a  $2^{O(n/d)}$ time algorithm for {\sc Degree Bounded Spanning Tree}?
\end{itemize}


\section*{Acknowledgements}
The first author thanks Fedor V. Fomin for hosting him at a visit
in the University of Bergen where this research was initiated.
\bibliography{multiset}
\bibliographystyle{abbrv}
\end{document}